\documentclass[conference,doublecolumn]{IEEEtran}

\newtheorem{thm}{Theorem}
\newtheorem{cor}{Corollary}
\newtheorem{lem}{Lemma}

\usepackage[final]{graphicx}
\usepackage[reqno]{amsmath}
\usepackage{amssymb}
\usepackage{subfig}
\usepackage{epstopdf}

\begin{document}

\title{On Optimal Message Assignments for Interference Channels with CoMP Transmission}
\author{{\large{Aly El Gamal, V.~Sreekanth Annapureddy, and Venugopal V.~Veeravalli}}\\ \large{ECE Department and Coordinated Science Laboratory}\\\large{University of Illinois at Urbana-Champaign}}

\maketitle

\begin{abstract}
The degrees of freedom (DoF) number of the fully connected $K-$user Gaussian interference channel is known to be $\frac{K}{2}$ (see~\cite{Cadambe-IA}). In~\cite{Annapureddy-ElGamal-Veeravalli-IT11}, the DoF for the same channel model was studied while allowing each message to be available at its own transmitter as well as $M-1$ successive transmitters. In particular, it was shown that the DoF gain through cooperation does not scale with the number of users $K$ for a fixed value of $M$, i.e., the per user DoF number is $\frac{1}{2}$. In this work, we relax the cooperation constraint such that each message can be assigned to $M$ transmitters without imposing further constraints on their location. Under the new constraint, we study properties for different message assignments in terms of the gain in the per user DoF number over that achieved without cooperation. In particular, we show that a local cooperation constraint that confines the transmit set of each message within a $o(K)$ radius cannot achieve a per user DoF number that is greater than $\frac{1}{2}$. Moreover, we show that the same conclusion about the per user DoF number holds for any assignment of messages such that each message cannot be available at more than two transmitters. Finally, for the case where $M>2$, we do not know whether a per user DoF number that is greater than $\frac{1}{2}$ is achievable. However, we identify a candidate class of message assignments that could potentially lead to a positive answer.    
\end{abstract}

\section{Introduction}
As a result of developments in the infrastructure of cellular networks, there has been a recent growing interest in the potential of {\em cooperative} transmission techniques where, through a backhaul link, messages can be available at more than one transmitter, i.e., Coordinated Multi-Point (CoMP) transmission. This new development has a proven advantage~(see, e.g.,~\cite{CoMP-book}) for mitigating the effect of interfering signals, in particular, for cell-edge users. We formally pose the problem of maximizing the sum rate in a multiuser channel with CoMP transmission by defining a \emph{cooperation order} constraint, which bounds the maximum number of transmitters at which any message can be available by a cooperation order $M$. 

In this work, we consider the fully connected Gaussian interference channel model with generic channel coefficients (as defined in~\cite{Cadambe-IA}), and study the degrees of freedom (DoF) of the channel while allowing each message to be available at $M$ transmitters.  
The DoF criterion provides an analytically tractable way to characterize the sum capacity and captures the number of \emph{interference-free} sessions in the channel.  In~\cite{Madsen-Nosratinia}, the DoF number for the $K-$user channel was shown to be upper bounded by $K/2$, i.e., the per user DoF number is upper bounded by $1/2$. This was shown to be achievable through the interference alignment (IA) scheme in~\cite{Cadambe-IA}.

In~\cite{Annapureddy-ElGamal-Veeravalli-IT11}, the DoF of a $K-$user channel was studied in the special case where each message is assigned to its own transmitter as well as $M-1$ successive transmitters. In particular, it was shown that the DoF \emph{gain} (over $K/2$) achieved through this kind of cooperation does not scale with $K$ for a fixed value of $M$. We shed more light on this issue in this work. More precisely, we study the asymptotic per user DoF number as the number of users in the network increases, and message assignments that may lead to a DoF gain that scales with $K$. 


\section{System Model and Notations}\label{sec:systemmodel}
We use the standard model for the $K-$user interference channel with single antenna transmitters and receivers.
\begin{equation}
Y_i(t) = \sum_{j=1}^{K} H_{ij}(t) X_j(t) + Z_i(t)
\end{equation}
where $t$ is the time index, $X_i(t)$ is the transmitted signal at transmitter $i$, $Y_i(t)$ is the received signal at receiver $i$, $Z_i(t)$ is the zero mean unit variance Gaussian noise at receiver $i$, and $H_{ij} (t)$ is the channel coefficient from transmitter $j$ to receiver $i$ over the $t^{th}$ time slot. We assume that the channel coefficients are drawn independently from a continuous distribution. i.e., the channel coefficients are \emph{generic}. For each $i \in [K]$, let $W_i$ be the message intended for receiver $i$.

We use $[K]$ to denote the set $\{1,2,\ldots,K\}$. For any set ${\cal A} \subseteq [K]$, we define the complement set $\bar{\cal A} = \{i: i\in[K], i\notin {\cal A}\}$. We use the abbreviations $W_{\cal A}$, $X_{\cal A}$, and $Y_{\cal A}$ to denote the sets $\{W_i, i\in {\cal A}\}$, $\{X_i, i\in {\cal A}\}$, and $\{Y_i, i\in {\cal A}\}$, respectively.

\subsection{Cooperation Model}
Let ${\cal T}_i \subseteq [K]$ be the transmit set of receiver $i$, i.e., those transmitters with the knowledge of $W_i$. The transmitters in ${\cal T}_i$ cooperatively transmit the message $W_i$ to the receiver $i$. The messages $\{W_i\}$ are assumed to be independent of each other. The cooperation order $M$ is defined as the maximum size of a transmit set:
\begin{equation}\label{eq:coop_order}
M = \max_i |{\cal T}_i|.
\end{equation}

For any set ${\cal S} \subseteq [K]$, we define $C_{\cal S}$ as the set of messages carried by transmitters with indices in ${\cal S}$, i.e., the set $\{i: {\cal T}_i \cap {\cal S} \neq \phi\}$.
\subsection{Degrees of Freedom}
The total power constraint across all the users is $P$.  The rates $R_i(P) = \frac{\log|W_i|}{n}$ are achievable if the decoding error probabilities of all messages can be simultaneously made arbitrarily small for large enough $n$. The capacity region $\mathcal{C}(P)$ is the set of all achievable rate tuples. The total number of degrees of freedom ($\eta$) is defined as $\limsup_{P \rightarrow \infty}\frac{ C_{\Sigma}(P)}{\log P}$, where $C_\Sigma(P)$ is the sum capacity. Since $\eta$ depends on the specific choice of transmit sets as well as the realization of the channel coefficients, we define $\eta(K,M)$ as the best achievable $\eta$ over all choices of transmit sets satisfying the cooperation order constraint in \eqref{eq:coop_order}, that holds for almost all realizations of a $K-$user channel defined as above.
In order to simplify our analysis, we define the following criterion to measure how $\eta(K,M)$ scales with $K$ for a fixed $M$.
\begin{equation}\label{eq:tau}
\tau(M) = \lim_{K\rightarrow \infty} \frac{\eta(K,M)}{K}
\end{equation}
It is worth noting here that the bounds derived in~\cite{Madsen-Nosratinia} and~\cite{Cadambe-IA} imply that $\tau(1)=\frac{1}{2}$.

\subsection{Message Assignment Strategy}
A message assignment strategy is defined by a sequence of transmit sets $({\cal T}_{i,K}), i\in[K], K\in\{1,2,\ldots\}$, where for each positive integer $K$, ${\cal T}_{i,K} \subseteq [K], |{\cal T}_{i,K}| \leq M, \forall i\in[K]$. We call a message assignment strategy \emph{optimal} for a sequence of $K-$user channels defined as above, $K\in\{1,2,\ldots\}$,  if and only if there exists a sequence of coding schemes achieving $\tau(M)$ where for any positive integer $K$, the transmit sets $({\cal T}_i)_{i\in[K]}$ used for the $K-$user channel are the transmit sets $({\cal T}_{i,K})_{i\in[K]}$ defined by the strategy.

\section{DoF Upper Bound}
In order to characterize the DoF of the channel $\tau(M)$, we need to consider all possible strategies for message assignments satisfying the cooperation order constraint defined in~\eqref{eq:coop_order}. In this section, we provide a way to upper bound the maximum achievable DoF for each such assignment, thereby, introducing a criterion for comparing different message assignments satisfying~\eqref{eq:coop_order} using the special cases where this bound holds tightly.

We start by stating the following auxiliary lemma for any $K$-user Gaussian interference channel with a DoF number of $\eta$. For any set ${\cal A} \subseteq [K]$, Define $U_{\cal A} = \cup_{i \notin {\cal A}} {\cal T}_i$, then, 
\begin{lem}\label{lem:dofouterbound}(\cite{ElGamal-Annapureddy-Veeravalli-ICC12}, Lemma $2$)
If there exists a set ${\cal A}\subseteq [K]$ and a function $f$, such that $f\left(Y_{\cal A},Z_{\cal A},X_{\bar{U}_{\cal A}}\right)=X_{U_{\cal A}}$, then $\eta \leq |{\cal A}|$. 
\end{lem} 
\begin{proof}
The proof is available in~\cite{ElGamal-Annapureddy-Veeravalli-ICC12}.
\end{proof}
Now, we prove the following corollary.
\begin{cor}\label{cor:dofouterbound}
For any $m,\bar{m} : m+\bar{m} \geq K$, if there exists a set ${\cal S}$ of indices for transmitters carrying no more than $m$ messages, and $|{\cal S}|=\bar{m}$, then $\eta \leq m$, or more precisely, 
\begin{equation}
\eta \leq \min_{{\cal S} \subseteq [K]}\max (|C_{\cal S}|, K-|{\cal S}|).
\end{equation} 
\end{cor}
\begin{proof}
We apply Lemma~\ref{lem:dofouterbound} with the set ${\cal A}$ defined as follows: 

Initially, set ${\cal A}$ as the set of indices for messages carried by transmitters with indices in ${\cal S}$. i.e., ${\cal A}=C_{\cal S}$. Now, if $|{\cal A}| <  K-|{\cal S}|$, then augment the set ${\cal A}$ with arbitrary message indices such that $|{\cal A}|=K-|{\cal S}|$.

We now note that the above construction guarantees that $|{\cal A}|+|{\cal S}| \geq K$ and that $U_{\cal A} \subseteq \bar{\cal S}$, hence, using Lemma~\ref{lem:dofouterbound}, it suffices to show the existence of a function $f$ such that $f(Y_{\cal A},Z_{\cal A},X_{\cal S})=X_{\bar{\cal S}}$. 

Since the channel is fully connected, by removing the Gaussian noise signals $Z_{\cal A}$ and transmit signals in $X_{\cal S}$ from received signals in $Y_{\cal A}$, we obtain at least $K-|{\cal S}|=|\bar{\cal S}|$ linear equations in the transmit signals in $X_{\bar{\cal S}}$. Moreover, since the channel coefficients are generic, those equations will be linearly independent with high probability, and hence, we can reconstruct $X_{\bar{\cal S}}$ from $|\bar{\cal S}|$ linearly independent equations.
\end{proof}
Please refer to Figure~\ref{fig:dofouterbound} for an example illustration of Lemma~\ref{lem:dofouterbound}.
\begin{figure}[htb]
\centering
\includegraphics[width=0.5\columnwidth]{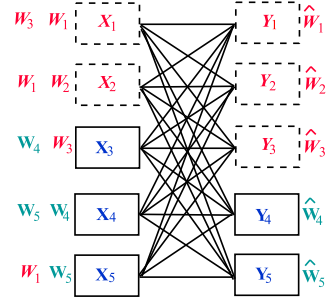}
\caption{Example application of Lemma~\ref{lem:dofouterbound}. ${\cal S}=\{1,2\}$, $C_{\cal S}=\{1,2,3\}$. Transmit signals with indices in ${\cal S}$, and messages as well as receive signals with indices in $C_{\cal S}$ are shown in tilted red font and dashed boxes. $\eta \leq |C_{\cal S}|=K-|{\cal S}|=3$.}
\label{fig:dofouterbound}
\end{figure}
\section{Local Cooperation}
In an attempt to reduce the complexity of the problem of finding an optimal message assignment strategy, we begin by considering in this section, message assignment strategies satisfying a local cooperation constraint, that is, for some $r(K)=o(K)$,
\begin{equation}
{\cal T}_{i,K} \subseteq \{i-r(K),i-r(K)+1,\ldots,i+r(K)\}, \forall i \in [K], \forall K\in {\bf Z}^{+}
\end{equation}  
Let $\tau_{\scriptscriptstyle \mathrm{LOC}}(M)$ be the maximum achievable asymptotic per user DoF number $\tau(M)$ under the additional local cooperation constraint, then, 

\begin{thm}
\begin{equation}
\tau_{\scriptscriptstyle \mathrm{LOC}}(M) = \frac{1}{2}, \text{ for all } M.
\end{equation}
\end{thm}
\begin{proof}
Fix $M \in {\bf Z}^{+}$. For any value of $K\in {\bf Z}^{+}$, we use Corollary~\ref{cor:dofouterbound} with the set ${\cal S}=\{1,2,\ldots,\lceil \frac{K}{2} \rceil \}$. Note that $C_{\cal S} \subseteq \{1,2,\ldots, \lceil \frac{K}{2} \rceil + r(K)\}$, and hence, it follows that $\eta(K,M) \leq \lceil \frac{K}{2} \rceil + r(K)$. Finally, $\tau(M) = \lim_{K \rightarrow \infty} \frac{\eta(K,M)}{K} \leq \frac{1}{2}$. The lower bound follows from~\cite{Cadambe-IA} without cooperation.
\end{proof}
\section{Asymptotic DoF Cooperation Gain}
In this section, we investigate if it is possible for the cooperation gain to scale linearly with $K$ for fixed $M$. More precisely, we try to investigate whether $\tau(M) > \tau(1)=\frac{1}{2}$ for values of $M > 1$. In the last section, we showed that such a gain is not possible for message assignment strategies that satisfy a local cooperation constraint. At the end of this section, we prove a property for message assignment strategies that may lead to a value of $\tau(M) > \frac{1}{2}$ for $M>2$. We start by proving the following upper bound on $\tau(M)$ that is tight enough for finding $\tau(2)$.
\begin{thm}\label{thm:tauouterbound}
\begin{equation}
\tau(M) \leq \frac{M-1}{M}
\end{equation}
\end{thm}
Before proving the above Theorem, we need the following auxiliary lemmas,
\begin{lem}\label{lem:basis}
\begin{equation}
\text{There exists } i \in [K] \text{ such that } |C_{\{i\}}| \leq M.
\end{equation}
\end{lem}
\begin{proof}
The statement follows by the pigeonhole principle, since
\begin{equation}
\sum_{i=1}^{K} |C_{\{i\}}| = \sum_{i=1}^{K} |{\cal T}_i| \leq MK
\end{equation} 
\end{proof}
\begin{lem}\label{lem:inductionstep}
For $M\geq 2$, if $\exists {\cal A} \subset [K]$ such that $|{\cal A}|=n < K$, and $|C_{\cal A}| \leq (M-1)n+1$, then $\exists {\cal B} \subseteq [K]$ such that $|{\cal B}|=n+1$, and $|C_{\cal B}| \leq (M-1)(n+1)+1$. 
\end{lem}
\begin{proof}
We only consider the case where $K > (M-1)(n+1)+1$, as otherwise, the statement trivially holds. In this case, we can show that,
\begin{equation}\label{eq:inequalityone}
M(K-|C_{\cal A}|) < (K-n)((M-1)(n+1)+2-|C_{\cal A}|)
\end{equation}
The proof of ~\eqref{eq:inequalityone} is available in the Appendix. Note that the left hand side in the above equation is the maximum number of message instances for messages outside the set $C_{\cal A}$, i.e.,
\begin{eqnarray}
\sum_{i\in[K], i \notin {\cal A}} |C_{\{i\}}| &\leq& M(K-|C_{\cal A}|)\nonumber
\\&<&(K-n)((M-1)(n+1)+2-|C_{\cal A}|)\nonumber
\\
\end{eqnarray}
Since the number of transmitters outside the set ${\cal A}$ is $K-n$, it follows by the pigeonhole principle that there exists a transmitter whose index is outside ${\cal A}$ and carries at most $(M-1)(n+1)+1-|C_{\cal A}|$ messages whose indices are outside $C_{\cal A}$. More precisely,
\begin{equation}
\exists i \in [K]\backslash{\cal A}: |C_{\{i\}}\backslash{C_{\cal A}}| \leq (M-1)(n+1)+1-|C_{\cal A}|
\end{equation}
It follows that there exists a transmitter whose index is outside the set ${\cal A}$ and can be added to the set ${\cal A}$ to form the set ${\cal B}$ that satisfies the statement.  
\end{proof}
We now prove Theorem~\ref{thm:tauouterbound}. In particular, we show that the following lemma holds.
\begin{lem}
For $M\geq 2$,
\begin{equation}
\eta(K,M) \leq \frac{K(M-1)+M+1}{M}  
\end{equation}
\begin{proof}
Assume that $n=\frac{K-1}{M}$ is an integer. We know by induction from lemmas~\ref{lem:basis} and~\ref{lem:inductionstep} that $\exists S \subset [K]$, $|S|=n$, $|C_{\cal S}| \leq (M-1)n+1=\frac{K(M-1)+1}{M}=K-|{\cal S}|$. Now, applying Corollary~\ref{cor:dofouterbound} proves that $\eta(K,M) \leq \frac{K(M-1)+1}{M}$. For the case where $\frac{K-1}{M}$ is not an integer, let $x$ be the largest integer less than $K$ such that $\frac{x-1}{M}$ is an integer. Now, we ignore the last $K-x$ users and bound the sum DoF for the remaining users by $\frac{x(M-1)+1}{M}$ to show that $\eta(K,M) \leq \frac{x(M-1)+1}{M}+(K-x)$, and hence,
\begin{eqnarray}
\eta(K,M) &\leq& \frac{x(M-1)+1}{M}+(K-x)\nonumber
\\&=& \frac{K(M-1)+1}{M}+\frac{K-x}{M}\nonumber
\\&\leq& \frac{K(M-1)+1}{M}+1
\end{eqnarray}  
\end{proof}
\end{lem}
Together with the achievability result in~\cite{Cadambe-IA}, The statement in Theorem~\ref{thm:tauouterbound} implies the following corollary.
\begin{cor}
\begin{equation}
\tau(2) = \frac{1}{2}
\end{equation}
\end{cor}
The characterization of $\tau(M)$ for values of $M > 2$ remains an open question, as Theorem~\ref{thm:tauouterbound} is only an upper bound. Moreover, the following result shows that this upper bound is loose for $M=3$. 
\begin{thm}
\begin{equation}
\tau(3) \leq \frac{5}{8}
\end{equation}
\end{thm}
\begin{proof}
We prove the statement by induction, and in order to do so, we use Lemma~\ref{lem:basis} to provide the basis, and for the induction step, we use Lemma~\ref{lem:inductionstep} together with the following lemma.
\begin{lem}\label{lem:mthreeinductionstep}
For $M=3$, If $\exists {\cal A} \subset [K]$ such that $|{\cal A}|=n$, and $\frac{K+1}{4} \leq n < K$, $|C_{\cal A}| \leq n + \frac{K+1}{4}+ 1$, then $\exists {\cal B} \subset [K]$ such that $|{\cal B}|=n+1$, $|C_{\cal B}| \leq n + \frac{K+1}{4} + 2$.
\end{lem}

The proof of the above Lemma follows in a similar fashion to that of Lemma~\ref{lem:inductionstep}. Let $x=n+\frac{K+1}{4} + 1$. We only consider the case where $K > x+1$, as otherwise, the proof is trivial. We first assume the following,
\begin{equation}\label{eq:inequalitytwo}
3(K-|C_{\cal A}|) < (K-n)\left(n+\frac{K+1}{4}+3-|C_{\cal A}|\right)
\end{equation}
Now, it follows that,
\begin{eqnarray}
\sum_{i\in[K], i \notin {\cal A}} |C_{\{i\}}| &\leq& M(K-|C_{\cal A}|)\nonumber
\\&<&(K-n)\left(n+\frac{K+1}{4}+3-|C_{\cal A}|\right),\nonumber
\\
\end{eqnarray}
and hence,
\begin{equation}
\exists i \in [K]\backslash{\cal A}: |C_{\{i\}}\backslash{C_{\cal A}}| \leq n+\frac{K+1}{4}+2-|C_{\cal A}|,
\end{equation}
and then the set ${\cal B}={\cal A}\cup \{i\}$ satisfies the statement of the lemma. Finally, we need to show that~\eqref{eq:inequalitytwo} is true. For the case where $|C_{\cal A}|=x$,
\begin{eqnarray}
3x &=& \frac{3K}{4} + \frac{15}{4} + 3n\nonumber
\\&=& (2n+K) + (n - \frac{K}{4} + \frac{15}{4})\nonumber
\\&>& 2n+K,
\end{eqnarray}
and hence, $3(K-x) < 2(K-n)$, which implies~\eqref{eq:inequalitytwo} for the case where $|C_{\cal A}|=x$. Moreover, we note that each decrement of $|C_{\cal A}|$ increases the left hand side of~\eqref{eq:inequalitytwo} by $3$ and the right hand side by $(K-n)$, and we know that,
\begin{eqnarray}
K &>& x+1\nonumber
\\&=& n+\frac{K+1}{4}+2\nonumber
\\&\geq& n+2,
\end{eqnarray}
and hence, $K-n \geq 3$, so there is no loss of generality in assuming that $|C_{\cal A}|=x$ in the proof of~\eqref{eq:inequalitytwo}, and the statement of Lemma~\ref{lem:mthreeinductionstep} holds.

Now, we show that $\tau(3)=\lim_{K \rightarrow \infty} \frac{\eta(K,3)}{K} \leq \frac{5}{8}$. It suffices to show that $\eta(K,3) \leq \frac{5K}{8}+o(K)$ for all values of $K$ such that $\frac{K+1}{4}$ is an even positive integer, and hence, we make that assumption for $K$. Define the following,
\begin{equation}
x_1 = \frac{K+1}{4}
\end{equation}
\begin{equation}
x_2=\frac{K-7}{8}
\end{equation}
\begin{equation}
x_3 = 2 x_1 + 1 + x_2
\end{equation}
Now, we note that,
\begin{equation}\label{eq:zequality}
x_3 = K -  (x_1 + x_2) ,
\end{equation}
and by induction, it follows from lemmas~\ref{lem:basis} and~\ref{lem:inductionstep} that $\exists {\cal S}_1 \subset [K]$, $|{\cal S}_1|= x_1$, $|C_{{\cal S}_1}| \leq 2x_1+1$. We now apply induction again with the set ${\cal S}_1$ as a basis and use Lemma~\ref{lem:mthreeinductionstep} for the induction step to show that $\exists {\cal S}_2 \subset [K]$, $|{\cal S}_2|= x_1+x_2$, $|C_{{\cal S}_2}| \leq x_3=K-|{\cal S}_2|$, and hence, we get the following upper bound using Corollary~\ref{cor:dofouterbound},
\begin{eqnarray}
\eta(K,3) &\leq& x_3\nonumber 
\\&=&\frac{5(K+1)}{8}
\end{eqnarray}

\end{proof}

Note that all the DoF upper bounding proofs used so far employ Corollary~\ref{cor:dofouterbound}. We now restrict our attention to upper bounds on $\tau(M)$ that follow by a direct application of Corollary~\ref{cor:dofouterbound}. More precisely, for a $K-$user channel defined as in Section~\ref{sec:systemmodel} with an assignment of the transmit sets $\{{\cal T}_i\}_{i \in [K]}$,  define $B(K,\{{\cal T}_i\})$ as the upper bound that follows by Corollary~\ref{cor:dofouterbound} for this channel, i.e.,
\begin{equation}\label{eq:upperboundeqn}
B(K,\{{\cal T}_i\})=\min_{{\cal S} \subseteq [K]}\max (|C_{\cal S}|, K-|{\cal S}|)
\end{equation}
Now, let $\eta_{\textrm{out}}(K,M)$ and $\tau_{\textrm{out}}(M)$ be the corresponding upper bounds that apply on $\eta(K,M)$ and $\tau(M)$. 
\begin{eqnarray}\label{eq:bareta}
\eta_{\textrm{out}}(K,M) &=& \max_{\{{\cal T}_i\}_{i \in [K]}: {\cal T}_i \subseteq [K],|{\cal T}_i| \leq M, \forall i \in [K]} B(K,\{{\cal T}_i\}),\nonumber
\\
\\\tau_{\textrm{out}}(M) &=& \lim_{K \rightarrow \infty} \frac{\eta_{\textrm{out}}(K,M)}{K}
\end{eqnarray}
All the facts that we stated above about $\tau(M)$ hold for $\tau_{\textrm{out}}(M)$, as all the upper bounding proofs follow by a direct application of Corollary~\ref{cor:dofouterbound}. We now identify a property for message assignment strategies, for which $\tau_{\textrm{out}}(M) > \frac{1}{2}, \forall M>2$. Note that this does not necessarily imply that $\tau(M)>\frac{1}{2}, \forall M>2$, but it provides some insight into whether this statement might be true.

For each possible message assignment, define a bipartite graph with partite sets of size $K$. Vertices in one of the partite sets represent transmitters, and vertices in the other set represent messages. There exists an edge between two vertices if and only if the corresponding message is available at the designated transmitter. We note that the cooperation order constraint implies that the maximum degree of nodes in one of the partite sets is bounded by $M$. We now observe that for any set ${\cal A}$ of transmitters, $C_{\cal A}=\{i: {\cal T}_i \cap {\cal A} \neq \phi\}$ is just the neighboring set $N_G({\cal A})$ in the corresponding bipartite graph $G$. Please refer to Figure~\ref{fig:bpgraph} for an illustration of the bipartite graph representation of message assignments.
\begin{figure}[htb]
\centering
\includegraphics[width=0.9\columnwidth]{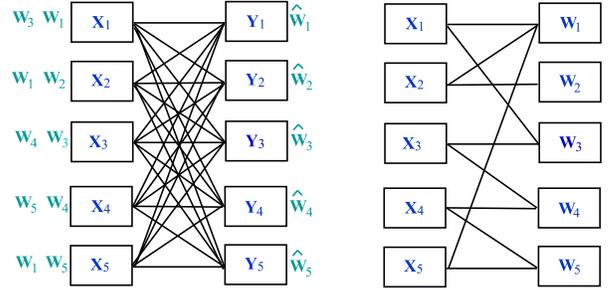}
\caption{The bipartite graph on the right side represents the message assignment for the $5-$user channel shown on the left side.}
\label{fig:bpgraph}
\end{figure}

Let ${\cal U}_G,{\cal V}_G,$ denote the partite sets corresponding to transmitters and messages in graph $G$, with respect to order. For all values of $i\in[K]$, define the following.
\begin{equation}
e_G(i) = \min_{{\cal A}\subseteq {\cal U}_G : |{\cal A}|=i} |N_G({\cal A})|
\end{equation}
then we can readily see that,
\begin{equation}
\eta_{\textrm{out}}(K,M) = \max_{G \in {\cal G}_M(K)} \min_{i \in [K]} \max(K-i,e_G(i))
\end{equation}
where ${\cal G}_M(K)$ is the set of all bipartite graphs, whose equi-sized partite sets have size $K$, and the maximum degree of the nodes in the partite set ${\cal V}_G$ is $M$.

For values of $M>2$, Pinsker proved the following result in $1973$~\cite{pinsker-1973}.
\begin{thm}\label{thm:pinsker}
For any $M>2$, $\exists$ a constant $c>1$,  and a sequence of $M-$regular bipartite graphs $(G_{M,K})$ whose partite sets have $K$ vertices, such that the following is true.
\begin{equation}\label{eq:expanders}
\lim_{K\rightarrow \infty} \frac{e_{G_{M,K}}(\alpha K)}{\alpha K} \geq c, \forall 0 < \alpha \leq \frac{1}{2}
\end{equation}
\end{thm}
We next show that the above statement implies that $\tau_{\textrm{out}}(M)>2, \forall M>2$.
\begin{cor}
\begin{equation}
\tau_{\textrm{out}}(M) > \frac{1}{2}, \forall M>2.
\end{equation}
\end{cor}
\begin{proof}
For each bipartite graph $G$ with partite sets of size $K$, define $i_{\textrm{min}}(G)$ as,
\begin{equation}
i_{\textrm{min}}(G)=\textrm{argmin}_i \max(K-i,e_G(i))
\end{equation}
Now, assume that $\tau_{\textrm{out}}(M) \leq \frac{1}{2}$, then for the sequence $(G_{M,K})$ chosen as in the statement of Theorem~\ref{thm:pinsker},   
\begin{equation}
\lim_{K\rightarrow \infty} \frac{\max(K-i_{\textrm{min}}(G_{M,K}), e_{G_{M,K}}(i_{\textrm{min}}(G_{M,K})))}{K} \leq \frac{1}{2}.
\end{equation}
It follows that,
\begin{equation}
\lim_{K\rightarrow \infty} \frac{K-i_{\textrm{min}}(G_{M,K})}{K} \leq \frac{1}{2},
\end{equation}
or,
\begin{equation}
\lim_{K\rightarrow \infty}  \frac{i_{\textrm{min}}(G_{M,K})}{K} \geq \frac{1}{2}
\end{equation}
But then as $e_G(i)$ is non-decreasing in $i$,~\eqref{eq:expanders} implies that,
\begin{equation}
\lim_{K\rightarrow \infty} \frac{e_{G_{M,K}}(i_{\textrm{min}}(G_{M,K}))}{K} > \frac{1}{2}.
\end{equation}
Therefore, the result in~\eqref{eq:expanders} implies that $\tau_{\textrm{out}}(M)>\frac{1}{2}, \forall M>2$.
\end{proof}
It is worth noting that a sequence of bipartite graphs satisfying~\eqref{eq:expanders} is said to define a \emph{vertex expander} as $K\rightarrow \infty$. To summarize, we have shown that for message assignment strategies corresponding to vertex expanders, one cannot apply the bound in~\eqref{eq:upperboundeqn} directly to show that $\tau(M)=\tau(1)=\frac{1}{2}$ for any $M>2$. Finally, we show that in case the upper bound $\tau_\textrm{out}(M)$ is tight, then using partial cooperation, the DoF gain can approach that achieved through assigning each message to all transmitters (full cooperation). More precisely, we show that,
\begin{thm}
\begin{equation}
\lim_{M \rightarrow \infty} \tau_{\textrm{out}}(M) = 1
\end{equation}
\end{thm}
\begin{proof}
The above statement follows as a corollary to a result in~\cite{appextremalcomb}. We provide the proof for completeness. We show that,
\begin{equation}\label{eq:limtau}
\forall \epsilon > 0 , \exists M(\epsilon): \forall M \geq M(\epsilon), {\tau}_{out}(M)> (1-\epsilon)
\end{equation}
For each positive integer $K$, we construct a bipartite graph $G_M(K)$, whose partite sets are of order $K$, by taking the union of $M$ {\em random} perfect matchings between the two partite sets. i.e., the matchings are probabilistically independent, and each is drawn uniformly from the set of all possible matchings. One can easily see that the maximum degree of nodes in $G_M(K)$ is bounded by $M$. i.e., $\Delta\left(G_M(K)\right) \leq M$, and hence, $G_M(K) \in {\cal G_M}(K)$. We will prove that for any $\epsilon > 0$,  there exists an $M(\epsilon)$ sufficiently large, such that for any $M \geq M(\epsilon)$, the probability that each set of $\epsilon K$ nodes in the partite set ${\cal U}_{G_M(K)}$ have more than $(1-\epsilon)K$ neighbors, is bounded away from zero for large enough $K$. More precisely, we show that,
\begin{equation}\label{eq:suffcondition}
\lim_{K \rightarrow \infty} Pr[\forall {\cal A} \subset {\cal U}_{G_M(K)}: |{\cal A}|=\epsilon K, |N_{G_M(K)}({\cal A})| > (1-\epsilon)K] > 0,
\end{equation}
and hence, for large enough $K$, there exists a graph $G$ in ${\cal G}_M(K)$ where all subsets of ${\cal U}_{G}$ of order $\epsilon K$ have more than $(1-\epsilon)K$ neighbors in ${\cal V}_G$, i.e.,
\begin{equation}
e_G(i) > (1-\epsilon)K, \forall i \geq \epsilon K
\end{equation} 
and it follows that $\eta_{\textrm{out}}(K,M) > (1-\epsilon)K$, and~\eqref{eq:limtau} holds. 

We now show that~\eqref{eq:suffcondition} holds. Let ${\cal A} \subset {\cal U}_{G_M(K)}, {\cal B} \subset {\cal V}_{G_M(K)}$ such that $|{\cal A}|=\epsilon K, |{\cal B}| = (1-\epsilon)K$. For any random perfect matching, the probability that all the neighbors of ${\cal A}$ are in ${\cal B}$ is $\frac{{(1-\epsilon)K\choose {\epsilon K}}}{{K \choose {\epsilon K}}}$. By independence of the matchings, we get the following,
\begin{eqnarray}
Pr[N_{G_M(K)}({\cal A}) \subseteq {\cal B}] &=& \left(\frac{{(1-\epsilon)K\choose {\epsilon K}}}{{K \choose {\epsilon K}}}\right)^M\nonumber
\\&\leq&\left((1-\epsilon)^{\epsilon K}\right)^M
\end{eqnarray}
A direct application of the union bound results in the following,
\begin{eqnarray}
Pr[|N_{G_M(K)}({\cal A})| &\leq& (1-\epsilon)K]\nonumber \\&\leq& \sum_{{\cal B} \subset {\cal V}_{G_M(K)}: |{\cal B}|=(1-\epsilon)K} Pr[N_{G_M(K)}({\cal A}) \subseteq {\cal B}]\nonumber
\\&\leq& {{K} \choose{(1-\epsilon)K}} (1-\epsilon)^{\epsilon M K}
\end{eqnarray}
and,
\begin{eqnarray}
Pr[\exists {\cal A} &\subset& {\cal U}_{G_M(K)}: |{\cal A}|=\epsilon K, |N_{G_M(K)}({\cal A})| \leq (1-\epsilon)K]\nonumber\\ &\leq& \sum_{{\cal A} \subset {\cal U}_{G_M(K)}: |{\cal A}|=\epsilon K} Pr[|N_{G_M(K)}({\cal A})| \leq (1-\epsilon)K]\nonumber
\\&\leq& {K \choose {\epsilon K}}{K \choose {(1-\epsilon)K}} (1-\epsilon)^{\epsilon M K}\nonumber
\\&=& {K \choose {\epsilon K}}^2 (1-\epsilon)^{\epsilon M K}\nonumber
\\&\overset{(a)}{\approx}& 2^{2K H(\epsilon)} (1-\epsilon)^{\epsilon M K}\nonumber
\\&=& 2^{\left(2H(\epsilon) + \epsilon M \log (1-\epsilon)\right)K}
\end{eqnarray}
where $H(.)$ is the binary entropy function, and $(a)$ follows as ${n \choose {\epsilon n}} \approx 2^{n H(\epsilon)}$ for large enough $n$. Now, we choose $M(\epsilon) > \frac{2H(\epsilon)}{- \epsilon \log(1-\epsilon)}$, to make the above exponent negative, and the above probability will be strictly less than unity, i.e., we showed that for any $M \geq M(\epsilon)$,
\begin{equation}
\lim_{K \rightarrow \infty} Pr[\exists {\cal A} \subset {\cal U}_{G_M(K)}: |{\cal A}|=\epsilon K, |N_{G_M(K)}({\cal A})| \leq (1-\epsilon)K] < 1,
\end{equation}
which implies that~\eqref{eq:suffcondition} is true.
\end{proof}
\section{Conclusion}
We studied the $K-$user fully connected Gaussian interference channel with a cooperation order constraint $M$. For the case where $M=2$, we showed that the limit of the per user DoF number $\tau(2)=\frac{1}{2}$ is the same as that achieved without cooperation. Moreover, for any value of $M$, we showed that message assignment strategies satisfying a local cooperation constraint cannot achieve a per user DoF number greater than $\frac{1}{2}$. Finally, we defined the upper bound on the per user DoF number $\tau_{\textrm{out}}(M)$ which characterizes all known upper bounds on $\tau(M)$, and showed that $\tau_{\textrm{out}}(M)>2, \forall M>2$, as a corollary of the existence of large bipartite vertex expanders. Thereby, suggesting that message assignment strategies corresponding to vertex expanders could potentially lead to a scalable DoF cooperation gain. 

\appendix
\section*{Auxiliary Lemma for Large Networks Upper Bounds}
\begin{lem}\label{lem:genineq}
If $K \geq (M-1)(n+1)+1$, $M \geq 2$, and $\exists {\cal S}\subseteq [K]$ such that $|{\cal S}| \leq (M-1)n+1$, then,
\begin{equation}
M(K-|{\cal S}|) < (K-n)\left((M-1)(n+1) + 2 - |{\cal S}|\right)
\end{equation}
\end{lem}
\begin{proof}
We first prove the statement for the case where $|{\cal S}| = (M-1)n+1$. This directly follows, as,
\begin{eqnarray}
M(K-|{\cal S}|) &=& M(K-((M-1)n+1))\nonumber
\\&\leq& M(K-(n+1))\nonumber
\\&<& M(K-n)\nonumber
\\&=&(K-n)\left((M-1)(n+1) + 2 - |{\cal S}|\right)\nonumber
\\
\end{eqnarray}
In order to complete the proof, we note that each decrement of $|{\cal S}|$ leads to an increase in the left hand side by $M$, and in the right hand side by $K-n$, and,
\begin{eqnarray}
K-n &\geq&  (M-1)(n+1)+1 - n\nonumber
\\&=&(M-2)n + M\nonumber
\\&\geq& M
\end{eqnarray} 
\end{proof}

\end{document}